\documentclass[a4paper,UKenglish,cleveref,autoref,thm-restate]{lipics-v2021}

\title{Bounded-Density Vertex Deletion Beyond Vertex Cover}
\titlerunning{Bounded-Density Vertex Deletion Beyond Vertex Cover}

\author{Adrian Micl\u au\c s}{Faculty of Mathematics and Computer Science, University of Bucharest, Romania}{adrian.miclaus1@gmail.com}{}{}
\author{Alexandru Popa}{Faculty of Mathematics and Computer Science, University of Bucharest, Romania}{alexandru.popa@fmi.unibuc.ro}{}{}
\author{Andrei Popa}{Faculty of Mathematics and Computer Science, University of Bucharest, Romania}{andreipopa0498@gmail.com}{}{}
\authorrunning{Adrian Micl\u{a}u\c{s}, Alexandru Popa, Andrei Popa}


\category{}
\relatedversion{}

\nolinenumbers

\EventEditors{}
\EventNoEds{0}
\EventLongTitle{}
\EventShortTitle{}
\EventAcronym{}
\EventYear{2027}
\EventDate{}
\EventLocation{}
\EventLogo{}
\SeriesVolume{0}
\ArticleNo{0}

\usepackage{amsmath}
\usepackage{amssymb}
\usepackage{mathtools}
\usepackage{booktabs}

\begin{document}
\maketitle

\begin{abstract}
We study \textsc{Bounded-Density Vertex Deletion} (BDVD),
where one asks whether at most \(k\) vertices can be deleted
so that every remaining subgraph has density at most a rational
input threshold \(\tau_\rho\).

We address the twin-cover and distance-to-cluster cases left
open by Raupach, Breitkopf, Herrmann, and Nichterlein
(MFCS 2026). We prove fixed-parameter tractability
parameterized solely by the twin-cover number, for unrestricted
input density. For distance to cluster, we establish
fixed-parameter tractability parameterized by
\(\operatorname{dc}(G)+\lfloor2\tau_\rho+1\rfloor\),
and complement this with \(\mathsf{W[1]}\)-hardness
parameterized by \(\operatorname{dc}(G)\) alone when the
target density is part of the input, even for integral targets.

Thus distance to cluster alone is fixed-parameter tractable
throughout \(\tau_\rho<1\), while it is insufficient for
unrestricted input density.

Finally, we settle the combined parameter treewidth plus
solution size, an open problem posed by Bazgan, Nichterlein,
and Vazquez Alferez (JCSS 2025). We prove that BDVD is
\(\mathsf{W[1]}\)-hard parameterized by
\(\operatorname{td}(G)+k\), and hence also by
\(\operatorname{tw}(G)+k\), already for targets strictly
between \(1\) and \(2\).
\end{abstract}

\newpage

\section{Introduction}

\subsection{Motivation}

The densest-subgraph problem asks for a subgraph containing many
edges relative to its number of vertices. For a nonempty graph
\(H\), its density is \(\rho(H)=|E(H)|/|V(H)|\), and the maximum
density of a subgraph of \(G\) is denoted by \(\rho^*(G)\).
The classical densest-subgraph problem is polynomial-time
solvable~\cite{Goldberg1984}, while efficient greedy approximation
algorithms have also been developed~\cite{Charikar2000}.
The problem and its numerous variants have been extensively
studied in the literature~\cite{LancianoMiyauchiFazzoneBonchi2024}.
In contrast to the unconstrained problem, imposing restrictions
on the size of the selected subgraph leads to additional
computational challenges~\cite{KhullerSaha2009}.
A natural complementary question is how difficult it is to
destroy all subgraphs whose density exceeds a prescribed threshold.

This leads to graph-modification problems in which a limited
number of vertices or edges may be removed in order to lower
the maximum subgraph density. Graph modification under
hereditary constraints is a classical topic in parameterized
complexity~\cite{Cai1996}. We study the vertex-deletion variant
\textsc{Bounded-Density Vertex Deletion} (BDVD), introduced
by Bazgan, Nichterlein, and Vazquez
Alferez~\cite{BazganNichterleinVazquez2025}.
The problem generalizes familiar deletion problems: at target
density \(0\), the residual graph must be independent, while
below density \(1\), every feasible residual graph must be
a forest. The difficulty is that feasibility is global:
after the deletion, every remaining subgraph must respect
the density threshold.

Structural parameterizations are therefore natural. If deleting
a small modulator exposes a highly regular remainder, the choices
outside the modulator may admit a compact description.
The parameters considered here form a natural hierarchy.
If \(X\) is a vertex cover, then \(G-X\) is independent.
If \(X\) is a twin cover, introduced by
Ganian~\cite{Ganian2015}, then \(G-X\) is a disjoint union
of cliques of true twins.
If \(X\) is a cluster deletion set, then \(G-X\) is an
arbitrary disjoint union of cliques; finding such deletion
sets has itself been extensively studied from a parameterized
perspective~\cite{HuffnerKomusiewiczMoserNiedermeier2010}.
Hence
\(\operatorname{dc}(G)\leq\operatorname{tc}(G)
\leq\operatorname{vc}(G)\).

The existing landscape left precisely the cases addressed
here unresolved.
Bazgan et al.~\cite{BazganNichterleinVazquez2025} proved
fixed-parameter tractability for vertex cover and asked
whether BDVD is fixed-parameter tractable for treewidth,
possibly together with the solution size.
Raupach, Breitkopf, Herrmann, and
Nichterlein~\cite{RaupachBreitkopfHerrmannNichterlein2026}
later proved hardness for treewidth and even treedepth,
obtained positive results for larger parameters such as
vertex integrity and max leaf number, and left twin cover
and distance to cluster open. They also reiterated the
open combined parameter treewidth plus solution size.

We close these gaps. Twin cover is sufficient for
fixed-parameter tractability even when the target density
is part of the input.
For distance to cluster, adding
\(\lfloor2\tau_\rho+1\rfloor\) yields an FPT algorithm,
but distance to cluster alone is \(\mathsf{W[1]}\)-hard
for unrestricted input density.
Both positive results exploit structural compression
combined with fixed-dimensional integer
programming~\cite{Lenstra1983}.
Finally, adding the deletion budget does not rescue
treewidth: the problem is already \(\mathsf{W[1]}\)-hard
for the stronger parameter treedepth plus solution size.

\subsection{Informal Problem Definition}

Let \(G=(V,E)\) be an undirected graph. For every nonempty subgraph \(H\subseteq G\), let \(\rho(H)=|E(H)|/|V(H)|\), and let \(\rho^*(G)=\max_{H\subseteq G}\rho(H)\). An instance of \textsc{Bounded-Density Vertex Deletion} consists of \(G\), an integer \(k\), and a rational target density \(\tau_\rho\). The question is whether there exists \(S\subseteq V(G)\) with \(|S|\leq k\) such that \(\rho^*(G-S)\leq\tau_\rho\).

Writing \(\tau_\rho=p/q\) with integers \(p\geq0\) and \(q>0\), the density condition is equivalent to \(q|E(H)|-p|V(H)|\leq0\) for every nonempty subgraph \(H\subseteq G-S\). Thus one may view every selected edge as contributing \(+q\) and every selected vertex as contributing \(-p\). This additive density-excess view is central both to our algorithms and to our hardness constructions.

\subsection{Previous and Related Work}

Bazgan, Nichterlein, and Vazquez Alferez~\cite{BazganNichterleinVazquez2025} initiated the systematic study of bounded-density deletion. For BDVD they obtained polynomial-time algorithms on trees and cliques, hardness on several restricted graph classes, fixed-parameter tractability parameterized by vertex cover, and hardness parameterized by the deletion budget. Their concluding questions asked, among other things, whether treewidth together with the solution size yields fixed-parameter tractability and what happens for structural parameters below vertex cover.

Raupach, Breitkopf, Herrmann, and Nichterlein~\cite{RaupachBreitkopfHerrmannNichterlein2026} refined this parameterized landscape. They proved \(\mathsf{W[1]}\)-hardness for parameters including treedepth and feedback vertex number, implying hardness for treewidth, while obtaining FPT algorithms for max leaf number and vertex integrity. If the target density is fixed, they also obtain tractability for clique-width. Their conclusion explicitly leaves the intermediate parameters distance to cluster and twin cover open, as well as the combined parameter treewidth plus solution size.

Twin cover was introduced by Ganian~\cite{Ganian2015} as a relaxation of vertex cover that allows clique components of true twins outside the modulator. For distance to cluster we use that a minimum cluster deletion set can be found in FPT time~\cite{HuffnerKomusiewiczMoserNiedermeier2010}. Both hardness reductions use unary multidimensional subset-sum variants proved \(\mathsf{W[1]}\)-hard by Ganian, Klute, and Ordyniak~\cite{GanianKluteOrdyniak2021}.

\subsection{Our Results}

\begin{table}[t]
\centering
\small
\begin{tabular}{lll}
\toprule
Parameter / restriction & Target density & Complexity \\
\midrule
Vertex cover & arbitrary input & FPT~\cite{BazganNichterleinVazquez2025} \\
Twin cover & arbitrary input & FPT, this work \\
Distance to cluster \(+\lfloor2\tau_\rho+1\rfloor\) & arbitrary input & FPT, this work \\
Distance to cluster & arbitrary input & \(\mathsf{W[1]}\)-hard, this work \\
Distance to cluster & \(\tau_\rho<1\) & FPT, this work \\
Treedepth \(+k\) & \(1<\tau_\rho<2\) & \(\mathsf{W[1]}\)-hard, this work \\
Treewidth \(+k\) & \(1<\tau_\rho<2\) & \(\mathsf{W[1]}\)-hard, this work \\
\bottomrule
\end{tabular}
\caption{The parameterized landscape established in this work, together with the known vertex-cover result.}
\label{tab:results}
\end{table}

Our first result closes the twin-cover case. The challenge beyond vertex cover is that \(G-X\) may contain arbitrarily many arbitrarily large cliques. We show that a worst density witness takes either all surviving vertices of a twin clique or none. A simultaneous level-filling argument then compresses all cliques with the same neighborhood in \(X\) to a single survivor count. The resulting integer program has dimension depending only on \(\operatorname{tc}(G)\), giving FPT for arbitrary input density.

Our second result gives an FPT algorithm for distance to cluster once the density-dependent quantity \(h=\lfloor2\tau_\rho+1\rfloor\) is included in the parameter. Unlike twin cliques, vertices in one cluster component may have different neighborhoods in the modulator. However, at most \(h\) vertices can survive in any clique. Truncating neighborhood multiplicities at \(h\) yields bounded signatures and configurations and again leads to a fixed-dimensional integer program. In particular, when \(\tau_\rho<1\), we have \(h\leq2\), so distance to cluster alone is FPT.

Our third result shows that this extra density dependence is essential in general. We prove \(\mathsf{W[1]}\)-hardness parameterized solely by distance to cluster when the target density is part of the input, even for integral targets and even when an explicit cluster deletion set is supplied. The reduction uses a clique gadget that behaves as a binary switch: an untouched item clique contributes its complete cross-edge load to the density excess, while one deletion makes every local contribution nonpositive.

Our fourth result settles the combined parameter treewidth plus solution size negatively. We prove the stronger \(\mathsf{W[1]}\)-hardness for \(\operatorname{td}(G)+k\), already for targets \(1<\tau_\rho<2\). The reduction encodes multidimensional capacities by multiplicities of length-two paths and includes a normalization argument showing that arbitrary deletions provide no advantage over deleting designated item vertices.

\section{Preliminaries}
\label{sec:preliminaries}

We consider finite, simple, undirected graphs. For a graph \(G\), write \(V(G)\) and \(E(G)\) for its vertex and edge sets, and \(n=|V(G)|\), \(m=|E(G)|\). For \(X\subseteq V(G)\), let \(G-X\) be obtained by deleting \(X\) and let \(G[X]\) be the subgraph induced by \(X\).

For every nonempty graph \(H\), its density is \(\rho(H)=|E(H)|/|V(H)|\). Since adding an edge between already selected vertices cannot decrease density, a densest subgraph may always be chosen induced, so \(\rho^*(G)=\max_{\emptyset\neq X\subseteq V(G)}|E(G[X])|/|X|\). We use \(\rho^*(\emptyset)=0\). For any integer representation \(\tau_\rho=p/q\), with \(p\geq0\) and \(q>0\), define the density excess \(\Phi(H)=q|E(H)|-p|V(H)|\). Then \(\rho(H)\leq p/q\) if and only if \(\Phi(H)\leq0\). Replacing \((p,q)\) by a common positive multiple only scales all excess values and therefore does not affect feasibility.

A vertex cover is a set \(X\) such that \(G-X\) is independent. Two adjacent vertices \(u,v\) are true twins if \(N[u]=N[v]\). A twin cover is a set \(X\) such that every edge with both endpoints outside \(X\) joins true twins~\cite{Ganian2015}; equivalently, every component of \(G-X\) is a clique whose vertices have the same neighborhood in \(X\). Its minimum size is \(\operatorname{tc}(G)\). A cluster deletion set is a set \(X\) such that \(G-X\) is a disjoint union of cliques, and its minimum size is the distance to cluster \(\operatorname{dc}(G)\). A minimum cluster deletion set can be found in FPT time~\cite{HuffnerKomusiewiczMoserNiedermeier2010}. These definitions imply \(\operatorname{dc}(G)\leq\operatorname{tc}(G)\leq\operatorname{vc}(G)\).

The treewidth \(\operatorname{tw}(G)\) is the minimum width of a tree decomposition of \(G\). The treedepth \(\operatorname{td}(G)\) is the minimum height of a rooted forest whose closure contains \(G\). In particular, \(\operatorname{tw}(G)\leq\operatorname{td}(G)-1\). Finally, the density of \(K_r\) is \((r-1)/2\), so every graph satisfying \(\rho^*(G)\leq\tau_\rho\) contains no clique on more than \(\lfloor2\tau_\rho+1\rfloor\) vertices.

\section{Fixed-Parameter Tractability for Twin Cover}
\label{sec:twin-cover}

Let \(X\) be a twin cover of size \(t\). The graph \(G-X\) may contain arbitrarily many arbitrarily large clique components, but all vertices inside one component are indistinguishable with respect to \(X\). We combine this symmetry with convexity to compress the choices outside \(X\) into a number of integer variables depending only on \(t\).

Let \(S\) be a deletion set. We first guess \(D=S\cap X\), discarding choices with \(|D|>k\), and write \(U=X\setminus D\). There are at most \(2^t\) branches.

The components of \(G-X\) are called \emph{twin cliques}. For every twin clique \(C\), let \(A_C=N(C)\cap X\) be its common neighborhood in \(X\), let \(n_C=|C|\), and let \(r_C\) denote its number of surviving vertices.

Fix \(Y\subseteq U\), representing the vertices of the twin cover selected by a potential density witness. Suppose that the witness selects exactly \(a\) of the \(r_C\) surviving vertices of \(C\), and let \(d_C(Y)=|A_C\cap Y|\). These vertices create \(\binom{a}{2}\) edges inside \(C\) and \(ad_C(Y)\) edges to \(Y\). Their contribution to the density excess is therefore \(\phi_{C,Y}(a)=q\binom{a}{2}+(qd_C(Y)-p)a\).

\begin{lemma}
\label{lem:twin-endpoints}
Let \(Y\subseteq U\), and let \(C\) be a twin clique with \(r_C\) surviving vertices. The maximum of \(\phi_{C,Y}(a)\) over integers \(0\leq a\leq r_C\) is attained at \(a=0\) or \(a=r_C\).
\end{lemma}

\begin{proof}
The function \(\phi_{C,Y}(a)=qa(a-1)/2+(qd_C(Y)-p)a\) is convex because \(q>0\). Its maximum on \([0,r_C]\), and hence on the integer points of this interval, is attained at an endpoint.
\end{proof}

For \(d\in\{0,\ldots,t\}\) and \(r\geq0\), define \(\psi_d(r)=\max\{0,q\binom{r}{2}+(qd-p)r\}\). By \cref{lem:twin-endpoints}, a worst witness either takes all surviving vertices of a twin clique or none. Hence fixed survivor counts are feasible precisely when \(q|E(G[Y])|-p|Y|+\sum_C\psi_{d_C(Y)}(r_C)\leq0\) for every \(Y\subseteq U\).

We next eliminate the need for one survivor count per clique. For every \(A\subseteq X\), let \(\mathcal C_A\) be the family of twin cliques with \(A_C=A\). We call \(A\) their \emph{neighborhood type}. There are at most \(2^t\) types. If the cliques of \(\mathcal C_A\) have survivor counts \(r_1,\ldots,r_m\), let \(R_A=\sum_i r_i\).

For a fixed witness \(Y\), all cliques of type \(A\) have the same value \(d=|A\cap Y|\), and their total contribution is \(\sum_i\psi_d(r_i)\). The distribution of the \(R_A\) survivors among the individual cliques therefore matters. Convexity, however, gives a canonical best distribution. Intuitively, concentrating many surviving vertices in one clique creates increasingly many internal edges, and hence makes that clique progressively more attractive to a violating density witness.

For example, suppose that one neighborhood type \(A\) contains three twin cliques of sizes \(5,4,\) and \(2\), and that \(R_A=6\) vertices are to survive among them. The distribution \((4,2,0)\) places four survivors in one clique, already creating six internal edges there. In contrast, level filling first places one survivor in each clique and then a second survivor in each clique, giving \((2,2,2)\). The latter spreads the internal edges as evenly as the clique capacities permit. More generally, the \(j\)-th survivor placed in a clique has a marginal cost depending only on the level \(j\) and the interaction value \(d\). Since these marginal costs are nondecreasing in \(j\), lower levels should always be filled before higher levels.

For \(d\in\{0,\ldots,t\}\), define the marginal cost of the \(j\)-th surviving vertex of a clique by \(\Delta_d(j)=\psi_d(j)-\psi_d(j-1)\).

\begin{lemma}
\label{lem:level-filling}
Fix a neighborhood type \(A\subseteq X\) and an integer \(R\). Among all distributions of \(R\) surviving vertices over the cliques of \(\mathcal C_A\), there exists one distribution that simultaneously minimizes \(\sum_i\psi_d(r_i)\) for every \(d\in\{0,\ldots,t\}\). It is obtained by filling the cliques level by level: before retaining a \((j+1)\)-st vertex in a clique, a \(j\)-th vertex is retained in every clique having capacity at least \(j\), except possibly at the final partially filled level.
\end{lemma}

\begin{proof}
Fix \(d\). The function \(r\mapsto q\binom{r}{2}+(qd-p)r\) is convex, and so is its maximum with the constant function \(0\). Hence \(\psi_d\) is convex and \(\Delta_d(1)\leq\Delta_d(2)\leq\cdots\).

View every clique as providing positions at levels \(1,2,\ldots\) up to its size. A position at level \(j\) has cost \(\Delta_d(j)\), and using such a position requires all lower-level positions of the same clique. Since costs are nondecreasing with the level, an optimum uses all available lower-level positions before higher ones, except possibly at the last partially filled level.

Crucially, although the values \(\Delta_d(j)\) depend on \(d\), their ordering by level does not. Thus the same level-filling distribution is simultaneously optimal for every \(d\).
\end{proof}

The simultaneous part of \cref{lem:level-filling} is essential. If the minimizing distribution depended on \(d=|A\cap Y|\), then different potential witnesses \(Y\) could require different distributions of the same total number \(R_A\) of survivors. In that case, the single quantity \(R_A\) would not suffice to represent a neighborhood type. The fact that all values of \(d\) induce the same ordering of levels is exactly what permits the compression to one integer variable per type.

The last sentence of the lemma is therefore the key compression step: the canonical distribution does not depend on which witness \(Y\) we later consider.

For \(A\subseteq X\), \(d\in\{0,\ldots,t\}\), and \(R\in\{0,\ldots,N_A\}\), where \(N_A=\sum_{C\in\mathcal C_A}|C|\), let \(C_{A,d}(R)\) be the minimum of \(\sum_{C\in\mathcal C_A}\psi_d(r_C)\) over all distributions with total survivor count \(R\). By \cref{lem:level-filling}, the same level-filled distribution realizes this minimum for every \(d\). Therefore the density constraint for \(Y\subseteq U\) becomes \(q|E(G[Y])|-p|Y|+\sum_{A\subseteq X}C_{A,|A\cap Y|}(R_A)\leq0\).

Thus arbitrarily many twin cliques have been compressed to at most \(2^t\) survivor counts \(R_A\).

It remains to represent the functions \(C_{A,d}\) linearly. For \(j\geq1\), let \(m_{A,j}\) be the number of cliques in \(\mathcal C_A\) having size at least \(j\). Define \(\gamma_{A,d}(1),\ldots,\gamma_{A,d}(N_A)\) by repeating \(\Delta_d(j)\) exactly \(m_{A,j}\) times for every nonempty level \(j\). Then \(\gamma_{A,d}\) is nondecreasing and \(C_{A,d}(R)=\sum_{j=1}^{R}\gamma_{A,d}(j)\).

Hence \(C_{A,d}\) is a discrete convex function: it is the cumulative sum of nondecreasing marginal costs. We represent it as the maximum of its supporting affine functions. For every pair \((A,d)\), introduce an integer variable \(z_{A,d}\), and for every \(s\in\{1,\ldots,N_A\}\) impose \(z_{A,d}\geq C_{A,d}(s-1)+\gamma_{A,d}(s)(R_A-s+1)\).

\begin{lemma}
\label{lem:convex-epigraph}
For every integral value of \(R_A\), the maximum of the right-hand sides of the preceding inequalities is exactly \(C_{A,d}(R_A)\).
\end{lemma}

\begin{proof}
Let \(R_A=r\). For \(s=r\), the right-hand side is \(C_{A,d}(r-1)+\gamma_{A,d}(r)=C_{A,d}(r)\). If \(s<r\), monotonicity of \(\gamma_{A,d}\) implies \(C_{A,d}(r)-C_{A,d}(s-1)\geq(r-s+1)\gamma_{A,d}(s)\). If \(s>r\), it implies \(C_{A,d}(s-1)-C_{A,d}(r)\leq(s-1-r)\gamma_{A,d}(s)\). Thus every other right-hand side is at most \(C_{A,d}(r)\).
\end{proof}

We now formulate an integer program for the fixed branch \(D\). For every neighborhood type \(A\subseteq X\), introduce an integer variable \(R_A\) with \(0\leq R_A\leq N_A\). For every \(A\subseteq X\) and \(d\in\{0,\ldots,t\}\), introduce \(z_{A,d}\) together with the epigraph inequalities above. For every \(Y\subseteq U\), impose \(q|E(G[Y])|-p|Y|+\sum_{A\subseteq X}z_{A,|A\cap Y|}\leq0\). Finally impose the budget constraint \(|D|+|V(G)\setminus X|-\sum_{A\subseteq X}R_A\leq k\).

Notice that \(Y=\emptyset\) is included. Hence the density constraints also detect violating subgraphs contained entirely in \(G-X\), without using any vertex of the twin cover.

\begin{lemma}
\label{lem:twin-ilp-correct}
For a fixed \(D\subseteq X\), the integer program is feasible if and only if there exists \(S\subseteq V(G)\) with \(S\cap X=D\), \(|S|\leq k\), and \(\rho^*(G-S)\leq p/q\).
\end{lemma}

\begin{proof}
Suppose first that the integer program is feasible. For every type \(A\), retain exactly \(R_A\) vertices according to the level-filling distribution and delete all remaining vertices outside \(X\), together with \(D\). The budget constraint guarantees at most \(k\) deletions.

Consider any induced subgraph \(H\) of the resulting graph and let \(Y=V(H)\cap U\). By \cref{lem:twin-endpoints}, for fixed \(Y\) the contribution of each twin clique is maximized by taking either all surviving vertices or none. Hence the maximum contribution of all cliques of type \(A\) is \(C_{A,|A\cap Y|}(R_A)\). By \cref{lem:convex-epigraph}, this is at most \(z_{A,|A\cap Y|}\). The density constraint for \(Y\) therefore implies \(\Phi(H)\leq0\).

Conversely, let \(S\) be a feasible deletion set with \(S\cap X=D\), and let \(R_A\) be the total number of surviving vertices of type \(A\). Redistribute these survivors according to level filling. By \cref{lem:level-filling}, no density constraint becomes worse. Set \(z_{A,d}=C_{A,d}(R_A)\). The epigraph constraints follow from \cref{lem:convex-epigraph}, the density constraints from feasibility, and the deletion budget is unchanged.
\end{proof}

There are at most \(2^t\) variables \(R_A\) and at most \((t+1)2^t\) variables \(z_{A,d}\), so the integer program has \(O(t2^t)\) integer variables. It has at most \(2^t\) density constraints and \(O(tn)\) epigraph inequalities. Its coefficients can be constructed in time polynomial in \(n,m,\log p,\log q\); the values \(p,q\) affect only the coefficients and not the dimension.

A minimum twin cover can itself be computed in fixed-parameter time~\cite{Ganian2015}. Together with fixed-dimensional integer linear programming~\cite{Lenstra1983}, this yields the following theorem.

\begin{theorem}
\label{thm:twin-cover}
\textsc{Bounded-Density Vertex Deletion} is fixed-parameter tractable parameterized by the twin-cover number. More precisely, for \(t=\operatorname{tc}(G)\) and \(\tau_\rho=p/q\), it can be solved in time \(g(t)\cdot\operatorname{poly}(n,m,\log p,\log q)\) for some computable function \(g\).
\end{theorem}

\begin{proof}
Compute a twin cover \(X\) of size \(t\) in FPT time and branch over the at most \(2^t\) choices for \(D=S\cap X\). By \cref{lem:twin-ilp-correct}, every branch is equivalent to an integer program with \(O(t2^t)\) integer variables and polynomially many constraints. Fixed-dimensional integer linear programming~\cite{Lenstra1983} solves every branch in FPT time with polynomial dependence on \(n,m,\log p,\log q\), and the factor \(2^t\) is absorbed into \(g(t)\).
\end{proof}


\section{Fixed-Parameter Tractability for Distance to Cluster}
\label{sec:distance-cluster}

Let \(X\subseteq V(G)\) be a cluster deletion set of size \(c\), so that \(G-X\) is a disjoint union of cliques. Unlike in the twin-cover setting, vertices of the same clique may have different neighborhoods in \(X\). The previous aggregation therefore does not apply.

A cluster deletion set of size \(c=\operatorname{dc}(G)\) can be found in FPT time parameterized by \(c\)~\cite{HuffnerKomusiewiczMoserNiedermeier2010}. We henceforth fix such a set \(X\).

Let \(\tau_\rho=p/q\) and define \(h=\lfloor2\tau_\rho+1\rfloor\). We guess \(D=S\cap X\), discard branches with \(|D|>k\), and write \(U=X\setminus D\). There are at most \(2^c\) branches.

The density target bounds how many vertices can survive in one clique.

\begin{lemma}
\label{lem:cluster-survivors}
Let \(C\) be a clique of \(G-X\). In every feasible solution, at most \(h=\lfloor2\tau_\rho+1\rfloor\) vertices of \(C\) survive.
\end{lemma}

\begin{proof}
If at least \(h+1\) vertices survived, they would induce \(K_{h+1}\), whose density is \(h/2\). Since \(h=\lfloor2\tau_\rho+1\rfloor>2\tau_\rho\), we obtain \(h/2>\tau_\rho\), contradicting feasibility.
\end{proof}

Thus an input clique may be arbitrarily large, but its residual part has size at most \(h\). This boundedness replaces the twin symmetry used in \cref{sec:twin-cover}.

Fix a clique \(C\) of \(G-X\). For every \(A\subseteq U\), let \(n_C(A)\) be the number of vertices \(v\in C\) satisfying \(N(v)\cap U=A\). Define the \emph{signature} \(\sigma_C\) by \(\sigma_C(A)=\min\{h,n_C(A)\}\). Since there are at most \(2^c\) subsets of \(U\), there are at most \((h+1)^{2^c}\) signatures.

A \emph{surviving configuration} compatible with \(\sigma\) is a vector \(s=(s_A)_{A\subseteq U}\) satisfying \(0\leq s_A\leq\sigma(A)\) for every \(A\) and \(\sum_A s_A\leq h\). The value \(s_A\) is the number of surviving vertices whose neighborhood in \(U\) is exactly \(A\).

Every feasible residual clique gives such a configuration, and every compatible configuration is realizable. Indeed, if \(\sigma(A)<h\), then \(n_C(A)=\sigma(A)\), while if \(\sigma(A)=h\), then \(n_C(A)\geq h\). Since at most \(h\) vertices survive in total, no larger multiplicity is relevant.

As a simple example, let \(U=\{x_1,x_2\}\) and \(h=3\). Suppose a clique \(C\) contains seven vertices with neighborhood \(\{x_1\}\), two with neighborhood \(\{x_2\}\), and one with neighborhood \(\{x_1,x_2\}\), and no vertices of the remaining type. Its signature records the multiplicities \(3,2,1,0\): the value \(7\) is truncated to \(3\) because no feasible solution can retain more than three vertices from \(C\) anyway. A possible surviving configuration is \(s_{\{x_1\}}=2\), \(s_{\{x_2\}}=1\), and all other entries zero. A potential density witness may then take, for instance, only one of the two surviving \(\{x_1\}\)-vertices, which is represented by a subconfiguration \(t\leq s\).

\begin{lemma}
\label{lem:cluster-counting}
For a fixed branch \(D\subseteq X\), the number of distinct signatures is at most \((h+1)^{2^c}\). Moreover, every signature admits at most \(\binom{h+2^c}{2^c}\) compatible surviving configurations, and every such configuration has a number of subconfigurations bounded by a function of \(c+h\).
\end{lemma}

\begin{proof}
There are at most \(2^c\) possible neighborhoods \(A\subseteq U\), and every signature entry belongs to \(\{0,\ldots,h\}\), giving at most \((h+1)^{2^c}\) signatures. A surviving configuration distributes at most \(h\) vertices among at most \(2^c\) neighborhood types. Ignoring the upper bounds \(s_A\leq\sigma(A)\) can only increase the count, and the number of nonnegative integer vectors with total at most \(h\) is \(\binom{h+2^c}{2^c}\). Finally, each coordinate of a subconfiguration \(t\leq s\) ranges over \(\{0,\ldots,s_A\}\), while \(\sum_A s_A\leq h\); hence the total number of subconfigurations is bounded solely in terms of \(c+h\).
\end{proof}

Fix \(W\subseteq U\), representing the modulator vertices selected by a potential witness, and let \(s\) be a surviving configuration. The witness may choose a subconfiguration \(t=(t_A)_{A\subseteq U}\) with \(0\leq t_A\leq s_A\), and we write \(|t|=\sum_A t_A\).

The selected clique vertices induce \(\binom{|t|}{2}\) edges, while a vertex of type \(A\) has \(|A\cap W|\) neighbors in \(W\). Their excess contribution is therefore \(g_W(t)=q\binom{|t|}{2}+q\sum_{A\subseteq U}|A\cap W|t_A-p|t|\).

For every \(W\subseteq U\) and configuration \(s\), let \(M_{W,s}\) be the maximum value of \(g_W(t)\) over all subconfigurations \(t\leq s\). Since \(|s|\leq h\), the number of such subconfigurations is bounded by a function of \(c+h\).

Different clique components of \(G-X\) are nonadjacent. Hence, after fixing \(W\), their maximum excess contributions can be optimized independently and added.

This independence is the reason the compression is compatible with a single global density constraint. Once \(W\) is fixed, choosing a maximizing subconfiguration in one clique neither creates edges to nor changes the contribution of another clique. Therefore the maximum excess of a witness using \(W\) is obtained by maximizing separately inside each residual clique and summing the resulting values \(M_{W,s}\).

For every signature \(\sigma\), let \(N_\sigma\) be the number of cliques with signature \(\sigma\). For every compatible configuration \(s\), introduce an integer variable \(x_{\sigma,s}\), representing the number of cliques of signature \(\sigma\) assigned configuration \(s\).

Every clique receives exactly one configuration, so for every \(\sigma\) we impose \(\sum_sx_{\sigma,s}=N_\sigma\). Writing \(n_X=|V(G)\setminus X|\), the deletion budget is \(|D|+n_X-\sum_{\sigma,s}|s|x_{\sigma,s}\leq k\).

Finally, for every \(W\subseteq U\), impose the density constraint \(q|E(G[W])|-p|W|+\sum_{\sigma,s}M_{W,s}x_{\sigma,s}\leq0\).

Again, \(W=\emptyset\) is included, so these constraints also detect violating subgraphs lying completely outside the modulator.

\begin{lemma}
\label{lem:cluster-ilp-correct}
For a fixed \(D\subseteq X\), the integer program above is feasible if and only if there exists \(S\subseteq V(G)\) with \(S\cap X=D\), \(|S|\leq k\), and \(\rho^*(G-S)\leq\tau_\rho\).
\end{lemma}

\begin{proof}
Suppose first that the integer program is feasible. For every \((\sigma,s)\), choose exactly \(x_{\sigma,s}\) cliques of signature \(\sigma\) and retain vertices realizing \(s\). Delete all other vertices outside \(X\), together with \(D\). The budget constraint guarantees at most \(k\) deletions.

Consider an induced subgraph \(H\) of the residual graph and let \(W=V(H)\cap U\). Inside each residual clique of configuration \(s\), the vertices selected by \(H\) form a subconfiguration whose excess is at most \(M_{W,s}\). Since distinct clique components are nonadjacent, the total excess of \(H\) is at most \(q|E(G[W])|-p|W|+\sum_{\sigma,s}M_{W,s}x_{\sigma,s}\), which is nonpositive by the constraint for \(W\).

Conversely, let \(S\) be feasible with \(S\cap X=D\). By \cref{lem:cluster-survivors}, every residual clique realizes a compatible configuration. Let \(x_{\sigma,s}\) count these configurations. The assignment and budget constraints follow immediately. For any \(W\subseteq U\), select in each residual clique a subconfiguration attaining \(M_{W,s}\). Their union with \(W\) has excess exactly \(q|E(G[W])|-p|W|+\sum_{\sigma,s}M_{W,s}x_{\sigma,s}\), which is nonpositive because \(S\) is feasible.
\end{proof}

By \cref{lem:cluster-counting}, the number of variables \(x_{\sigma,s}\) depends only on \(c+h\). There are at most \(2^c\) density constraints, and all signatures, configurations, subconfigurations, and values \(M_{W,s}\) can be generated in \(f(c+h)\cdot\operatorname{poly}(n,m,\log p,\log q)\) time.

\begin{theorem}
\label{thm:distance-cluster}
\textsc{Bounded-Density Vertex Deletion} is fixed-parameter tractable parameterized by \(\operatorname{dc}(G)+\lfloor2\tau_\rho+1\rfloor\). More precisely, for \(c=\operatorname{dc}(G)\) and \(h=\lfloor2\tau_\rho+1\rfloor\), it can be solved in \(f(c+h)\cdot\operatorname{poly}(n,m,\log p,\log q)\) time for some computable function \(f\).
\end{theorem}

\begin{proof}
Compute a cluster deletion set \(X\) of size \(c\) in FPT time and branch over the at most \(2^c\) possibilities for \(D=S\cap X\). For every branch, the integer program above has dimension bounded solely by a function of \(c+h\), and by \cref{lem:cluster-ilp-correct} it is feasible exactly when the branch admits a feasible solution. Fixed-dimensional integer linear programming~\cite{Lenstra1983} solves every branch in \(f_1(c+h)\cdot\operatorname{poly}(n,m,\log p,\log q)\) time.
\end{proof}


For targets below \(1\), the density-dependent part of the parameter becomes constant.

\begin{corollary}
\label{cor:distance-cluster-subunit}
For every input target density \(\tau_\rho<1\), \textsc{Bounded-Density Vertex Deletion} is fixed-parameter tractable parameterized solely by the distance to cluster.
\end{corollary}

\begin{proof}
If \(\tau_\rho<1\), then \(\lfloor2\tau_\rho+1\rfloor\leq2\). The claim follows from \cref{thm:distance-cluster}.
\end{proof}

\section{Hardness for Distance to Cluster Alone}
\label{sec:distance-cluster-hardness}

The dependence on the density threshold in \cref{thm:distance-cluster} cannot in general be removed. We prove that \textsc{Bounded-Density Vertex Deletion} is $\mathsf{W[1]}$-hard when parameterized solely by the distance to cluster. The reduction already uses an integral target density and supplies an explicit cluster deletion set.

We use the following exact-cardinality version of multidimensional subset sum.

\medskip
\noindent
\textsc{Unary Exact Multidimensional Subset Sum} (UEMSS)\\
\emph{Input:} vectors $a_1,\ldots,a_n\in\mathbb{N}^d$, a target vector $t\in\mathbb{N}^d$, and an integer $K$. All numbers are encoded in unary.\\
\emph{Question:} is there a set $R\subseteq[n]$ with $|R|=K$ and $\sum_{i\in R}a_i=t$?\\
\emph{Parameter:} $d$.

\begin{lemma}
\label{lem:uemss-hard}
\textsc{UEMSS} is $\mathsf{W[1]}$-hard parameterized by $d$, even on instances in which $K\leq d$.
\end{lemma}

\begin{proof}
Ganian, Klute, and Ordyniak~\cite{GanianKluteOrdyniak2021} prove $\mathsf{W[1]}$-hardness of multidimensional subset sum with unary-encoded coordinates by a reduction from \textsc{Multicolored Clique}. Inspecting their construction, the selector coordinates force every solution to contain exactly one vertex vector for each color class and exactly one edge vector for each pair of color classes. Hence the solution cardinality is fixed by the source parameter. Their construction uses $2\binom{r}{2}+r=r^2$ dimensions for a source instance with $r$ color classes, whereas every solution contains $\binom{r}{2}+r\leq r^2$ vectors. Thus the same reduction proves the stated restricted form.
\end{proof}

We next isolate the density gadget used by the reduction. Fix an integer $P\geq1$ and put $H=2P+1$. Let $C$ be a clique on $H$ vertices and let $X$ be an independent set with arbitrary edges between $C$ and $X$. For $W\subseteq X$ and $T\subseteq C$, define
$\psi_W(T)=\binom{|T|}{2}+|E(T,W)|-P|T|$.

\begin{lemma}
\label{lem:binary-clique-gadget}
Assume that $|E(C,X)|\leq P$. Then, for every $W\subseteq X$:
\begin{enumerate}
    \item if $T\subsetneq C$, then $\psi_W(T)\leq0$;
    \item $\psi_W(C)=|E(C,W)|$.
\end{enumerate}
Consequently, if all $H$ vertices of $C$ survive, its maximum local contribution to a density-excess witness using $W$ is $|E(C,W)|$, whereas after deleting at least one vertex of $C$ its maximum local contribution is zero.
\end{lemma}

\begin{proof}
The second claim follows from
$\binom{H}{2}-PH=0$.
Now let $1\leq s=|T|\leq H-1=2P$. Since $|E(T,W)|\leq |E(C,X)|\leq P$, we have $\psi_W(T)\leq\binom{s}{2}+P-Ps=P-s(2P-s+1)/2$.
For $1\leq s\leq2P$, the product $s(2P-s+1)$ is at least $2P$, and hence the right-hand side is nonpositive. The empty set has contribution zero. This proves both claims.
\end{proof}

\begin{theorem}
\label{thm:distance-cluster-hard}
\textsc{Bounded-Density Vertex Deletion} is $\mathsf{W[1]}$-hard parameterized by $\operatorname{dc}(G)$. This holds even if the target density is an integer and an explicit cluster deletion set $X$ is supplied with the instance.
\end{theorem}

\begin{proof}
We reduce from the restricted form of \textsc{UEMSS} in \cref{lem:uemss-hard}. Let $a_1,\ldots,a_n\in\mathbb{N}^d$, target $t\in\mathbb{N}^d$, and $K\leq d$ be an instance.

We first add one auxiliary coordinate and enough dummy items to obtain slack for the density gadget. Put $d'=d+1$. Extend every original vector by a final zero coordinate and extend the target by a final zero coordinate. Add $M=2d'(K+1)+K$ dummy vectors, each equal to the unit vector of the new coordinate. Denote the resulting vectors by $\widehat a_1,\ldots,\widehat a_N\in\mathbb{N}^{d'}$, where $N=n+M$, and let $\widehat t=(t,0)$. A size-$K$ solution cannot contain a dummy vector, so this padding preserves equivalence.

We encode the complement of a solution. Let $C=N-K=n+2d'(K+1)$ and $b=\sum_{i=1}^N\widehat a_i-\widehat t$. Thus $C\geq2d'(K+1)$. For every item $i\in[N]$ and coordinate $j\in[d']$, define $z_{i,j}=C\widehat a_i[j]-b[j]$. Choose $L=1+\max_{i,j}|z_{i,j}|$ and $P=CL$. Finally, define $w^+_{i,j}=L+z_{i,j}$ and $w^-_{i,j}=L-z_{i,j}$.
They satisfy $w^+_{i,j}+w^-_{i,j}=2L$. Moreover, since $|z_{i,j}|\leq L-1$, we have $1\leq w^\pm_{i,j}\leq2L-1<2P+1=H$, so all prescribed neighborhoods below are realizable in a simple graph.

We construct the BDVD graph $G$. The modulator $X$ consists, for every $j\in[d']$ and every $r\in[K+1]$, of two vertices $x^+_{j,r}$ and $x^-_{j,r}$. The set $X$ is independent and has $|X|=2d'(K+1)\leq2(d+1)^2$. For every item $i\in[N]$, create a clique $Q_i$ on $H=2P+1$ vertices. There are no edges between distinct item cliques. For every $j,r$, make $x^+_{j,r}$ adjacent to exactly $w^+_{i,j}$ vertices of $Q_i$, and make $x^-_{j,r}$ adjacent to exactly $w^-_{i,j}$ vertices of $Q_i$. These neighborhoods may be chosen arbitrarily subject to their prescribed cardinalities.

The total number of edges from one item clique to the modulator is $|E(Q_i,X)|=(K+1)\sum_{j=1}^{d'}(w^+_{i,j}+w^-_{i,j})=2d'(K+1)L\leq CL=P$.
Therefore \cref{lem:binary-clique-gadget} applies to every item clique. We set the target density to the integer
$\tau_\rho=P$
and the deletion budget to $K$.

Deleting $X$ leaves the disjoint union of the cliques $Q_1,\ldots,Q_N$. Hence $X$ is a cluster deletion set and
$\operatorname{dc}(G)\leq2d'(K+1)\leq2(d+1)^2$.
The construction is polynomial. Indeed, the source coordinates are unary encoded, all quantities above are polynomially bounded in the source input size, and every item clique has $2P+1$ vertices.

We record the effect of a deletion set. Let $Z$ be any set of at most $K$ deleted vertices and let $F=\{i\in[N]:Z\cap Q_i=\emptyset\}$ be the indices of item cliques that remain full. Fix $W\subseteq X\setminus Z$. Since $X$ is independent and distinct item cliques are nonadjacent, \cref{lem:binary-clique-gadget} shows that the maximum density excess among witnesses whose intersection with $X$ is exactly $W$ equals $-P|W|+\sum_{i\in F}|E(Q_i,W)|=\sum_{x\in W}(-P+\sum_{i\in F}|N(x)\cap Q_i|)$. It follows that $G-Z$ satisfies the target density if and only if
\begin{equation}
\label{eq:hard-singleton-capacity}
\sum_{i\in F}|N(x)\cap Q_i|\leq P
\end{equation}
for every surviving $x\in X$.

Suppose first that the UEMSS instance is a yes-instance, and let $R\subseteq[n]$ satisfy $|R|=K$ and $\sum_{i\in R}a_i=t$. Delete one arbitrary vertex from $Q_i$ for every $i\in R$. Then the set $F=[N]\setminus R$ has size $C$ and
$\sum_{i\in F}\widehat a_i=b$. Hence, for every $j\in[d']$, $\sum_{i\in F}z_{i,j}=C\sum_{i\in F}\widehat a_i[j]-Cb[j]=0$. Consequently, for every clone $x^+_{j,r}$ and $x^-_{j,r}$, $\sum_{i\in F}w^+_{i,j}=\sum_{i\in F}w^-_{i,j}=CL=P$.
Thus \eqref{eq:hard-singleton-capacity} holds, and the BDVD instance is a yes-instance.

Conversely, suppose that $Z$ is a feasible deletion set with $|Z|\leq K$. Let $R=[N]\setminus F$ be the set of item cliques hit by $Z$. Since every hit clique consumes at least one deletion,
$|R|\leq K$ and therefore $|F|\geq N-K=C$.
For each $j\in[d']$, the graph contains $K+1$ copies of both $x^+_{j,r}$ and $x^-_{j,r}$, while at most $K$ vertices are deleted in total. Hence at least one plus-copy and one minus-copy survive. Applying \eqref{eq:hard-singleton-capacity} to such surviving copies gives $|F|L+\sum_{i\in F}z_{i,j}\leq CL$ and $|F|L-\sum_{i\in F}z_{i,j}\leq CL$.
Adding the two inequalities yields $|F|\leq C$. Together with $|F|\geq C$, this implies $|F|=C$ and hence $|R|=K$.
In particular, all $K$ deletions are already needed to hit the $K$ cliques indexed by $R$, so no modulator vertex is deleted and exactly one vertex is deleted from every clique in $R$.

Since $|F|=C$, the two inequalities above further imply
$\sum_{i\in F}z_{i,j}=0$ for every $j$. By the definition of $z_{i,j}$, $0=C(\sum_{i\in F}\widehat a_i[j]-b[j])$, and therefore $\sum_{i\in F}\widehat a_i=b$. Taking complements gives $\sum_{i\in R}\widehat a_i=\sum_{i=1}^N\widehat a_i-b=\widehat t$.
The last coordinate of $\widehat t$ is zero, whereas every dummy vector has last coordinate one. Hence $R$ contains no dummy vector. Thus $R\subseteq[n]$, $|R|=K$, and $\sum_{i\in R}a_i=t$, so the original UEMSS instance is a yes-instance.

The reduction maps parameter $d$ to distance to cluster at most $2(d+1)^2$. This proves $\mathsf{W[1]}$-hardness.
\end{proof}


\section{Hardness for Treewidth Plus Solution Size}
\label{sec:tw-k-hardness}

We next settle the parameterization by treewidth together with the deletion budget.  In fact, we prove the stronger statement that the problem remains hard when treewidth is replaced by treedepth.  The reduction makes explicit why bounding the solution size does not remove the numerical obstruction caused by an input target density.

We use the following problem of Ganian, Klute, and Ordyniak~\cite{GanianKluteOrdyniak2021}.

An instance consists of vectors $a_1,\ldots,a_n\in\mathbb N^d$, a target vector $t\in\mathbb N^d$, and an integer $K$.  The question is whether there is a set $D\subseteq[n]$ with $|D|=K$ and $\sum_{i\in D}a_i=t$.  RMSS is $\mathsf{W[1]}$-hard parameterized by $d+K$, even when all integers are encoded in unary~\cite{GanianKluteOrdyniak2021}.

For a vertex set $U$ and a rational target $\tau_\rho=p/q$, define its \emph{excess} by
$\operatorname{exc}(U)=q|E(G[U])|-p|U|$.  Thus $\rho^*(G)\le p/q$ if and only if $\operatorname{exc}(U)\le0$ for every $U\subseteq V(G)$.

\begin{theorem}
\label{thm:tw-k-hard}
\textsc{Bounded-Density Vertex Deletion} is $\mathsf{W[1]}$-hard parameterized by $\operatorname{td}(G)+k$.  Consequently, it is $\mathsf{W[1]}$-hard parameterized by $\operatorname{tw}(G)+k$.  The hardness already holds for target densities $1<\tau_\rho<2$.
\end{theorem}

\begin{proof}
We reduce from unary RMSS.  Let $(a_1,\ldots,a_n,t,K)$ be an instance of dimension $d$.  If $K>n$, or if $K=n$ and $\sum_i a_i\ne t$, or if some coordinate $t[j]$ exceeds $A_j:=\sum_i a_i[j]$, we can output a fixed no-instance.  The case $K=n$ and $\sum_i a_i=t$ can similarly be mapped to a fixed yes-instance.  Hence assume $K<n$ and put
$C:=n-K\ge1$ and $B_j:=A_j-t[j]$ for $j\in[d]$.

For every item $i\in[n]$ and coordinate $j\in[d]$, let $z_{i,j}:=C\,a_i[j]-B_j$. Choose $L:=1+\max_{i,j}|z_{i,j}|$ and $P:=CL$, and define the positive integers $w^+_{i,j}:=L+z_{i,j}$ and $w^-_{i,j}:=L-z_{i,j}$.
We next encode the common capacity $P$ into the target density.  Let
\[
 p:=
 \begin{cases}
 3P+1,&P\text{ even},\\
 3P+2,&P\text{ odd},
 \end{cases}
 \qquad
 q:=\frac{p+3}{2}.
\]
Then $p=2q-3$, $q\ge4$, and
\begin{equation}
\label{eq:twk-gap}
 3P<p<3(P+1).
\end{equation}
In particular, $\tau_\rho:=p/q=2-3/q$ lies strictly between $1$ and $2$. We use this integer representation of the rational target throughout the excess calculation; reducing the fraction would only multiply all excess values by the same positive factor.

We construct the graph $G$ as follows.  For each signed coordinate $(j,\sigma)\in[d]\times\{+,-\}$, create $K+1$ \emph{anchor} vertices
$x^{1}_{j,\sigma},\ldots,x^{K+1}_{j,\sigma}$.

For every item $i\in[n]$, create an item core $F_i$ with two distinguished vertices $u_i,v_i$ and $q-2$ further vertices $h_i^1,\ldots,h_i^{q-2}$.  Add the edge $u_iv_i$ and, for every $r\in[q-2]$, the two edges $u_ih_i^r$ and $v_ih_i^r$.  Thus $F_i$ is $K_{2,q-2}$ together with the edge joining the two vertices on the side of size two.  It has exactly $q$ vertices and $2q-3=p$ edges.

Finally, for every item $i$, signed coordinate $(j,\sigma)$, anchor copy $r\in[K+1]$, and every $\ell\in[w^\sigma_{i,j}]$, create a connector vertex $y^{r,\ell}_{i,j,\sigma}$ adjacent precisely to $u_i$ and $x^r_{j,\sigma}$.  In other words, $w^\sigma_{i,j}$ internally vertex-disjoint paths of length two join $u_i$ to each copy of the corresponding anchor.  Set the BDVD deletion budget to $k:=K$ and the target density to $p/q$.

We first record the local excess properties of the construction. The whole item core has excess zero because $\operatorname{exc}(V(F_i))=q(2q-3)-pq=0$. Moreover, among subsets of $F_i$ containing $u_i$, the maximum internal excess is zero and is attained by the whole core. Indeed, if both $u_i,v_i$ and exactly $s$ of the $q-2$ common neighbors are present, then the excess is $q(2s+1)-p(s+2)=3(s-q+2)\leq0$,
with equality exactly for $s=q-2$.  If at most one of $u_i,v_i$ is present, the induced graph is a forest of density below $1<p/q$, and hence has negative excess whenever it is nonempty.

A connector vertex has an equally simple behavior. If both endpoints are in a candidate witness, inserting the connector changes the excess by $2q-p=3>0$.
If exactly one endpoint is present, the change is $q-p=3-q<0$, and if neither endpoint is present the change is $-p<0$.  Consequently, in a maximum-excess witness a connector is present exactly when both of its endpoints are present.

Consider first a deletion set $D\subseteq\{u_1,\ldots,u_n\}$ consisting only of item selectors, and let $F:=\{i:u_i\notin D\}$ be the surviving items.  By the preceding observations, every maximum-excess witness can be put in the following canonical form without decreasing its excess: whenever it contains $u_i$, it contains the whole core $F_i$; and for every chosen anchor it contains precisely all connector vertices whose other endpoint is a chosen surviving selector.  Since all connector multiplicities are nonnegative, once an anchor is chosen it is optimal to choose every surviving item selector.  Thus an anchor of signed coordinate $(j,\sigma)$ contributes
\begin{equation}
\label{eq:twk-anchor-excess}
 -p+3\sum_{i\in F}w^\sigma_{i,j}.
\end{equation}
Different anchors contribute additively, while the complete item cores contribute zero.  It follows that
\begin{equation}
\label{eq:twk-capacity}
 \rho^*(G-D)\le p/q
 \quad\Longleftrightarrow\quad
 \sum_{i\in F}w^\sigma_{i,j}\le P
 \text{ for every }j\in[d],\ \sigma\in\{+,-\},
\end{equation}
where we used \eqref{eq:twk-gap} and integrality of the left-hand side.

We next show that arbitrary vertex deletions do not provide any additional power.  Let $S$ be a feasible deletion set with $|S|\le K$.  For item $i$, call its \emph{private region} the vertices of $F_i$ together with all connector vertices incident with $u_i$.  Whenever $S$ intersects the private region of item $i$, replace all deletions made in that private region by the single deletion of $u_i$.  This does not increase the number of deleted vertices and preserves feasibility.  To see this, after deleting $u_i$ the vertices of $F_i-u_i$ form a star that is completely disconnected from the rest of the graph, while every restored connector is merely a pendant vertex at an anchor. Since $p/q>1$, every nonempty induced subgraph of the star has negative excess, and adding such a pendant connector to a witness changes the excess by $q-p<0$ if its anchor is present and by $-p<0$ otherwise. Thus all vertices restored in this private region can be removed from any positive-excess witness without decreasing its excess. The remaining witness was already present before the replacement, contradicting feasibility of the original deletion set.

After performing these replacements, the deletion set consists only of item selectors and anchors.  Anchor deletions can also be removed.  Indeed, there are $K+1$ identical copies of every signed-coordinate anchor, whereas at most $K$ vertices are deleted in total.  Hence at least one copy of every anchor type survives.  Feasibility therefore implies that the quantity in \eqref{eq:twk-anchor-excess} is nonpositive for every anchor type.  Restoring deleted copies merely adds further anchors with the same nonpositive contribution.  By the canonical-form argument above, feasibility is preserved.  Consequently every yes-instance admits a solution $D\subseteq\{u_1,\ldots,u_n\}$ of size at most $K$.

It remains to verify the arithmetic encoding.  Suppose first that the RMSS instance has a solution $D\subseteq[n]$ with $|D|=K$ and $\sum_{i\in D}a_i=t$.  Its complement $F=[n]\setminus D$ has $|F|=C$ and $\sum_{i\in F}a_i[j]=B_j$ for every coordinate $j$. Hence $\sum_{i\in F}z_{i,j}=C\sum_{i\in F}a_i[j]-|F|B_j=0$, and therefore $\sum_{i\in F}w^+_{i,j}=\sum_{i\in F}w^-_{i,j}=CL=P$.
and \eqref{eq:twk-capacity} shows that deleting the selectors $\{u_i:i\in D\}$ is a feasible BDVD solution.

Conversely, let $D$ be a canonical BDVD solution of size at most $K$, and let $F=[n]\setminus D$.  By \eqref{eq:twk-capacity}, for every $j$ we have $|F|L+\sum_{i\in F}z_{i,j}\leq CL$ and $|F|L-\sum_{i\in F}z_{i,j}\leq CL$.
Adding the inequalities gives $|F|\le C$.  On the other hand, $|D|\le K$ implies $|F|\ge n-K=C$.  Thus $|F|=C$ and $|D|=K$.  Substituting $|F|=C$ back into the two inequalities forces
$\sum_{i\in F}z_{i,j}=0$ for every $j$.  By the definition of $z_{i,j}$, $0=C\sum_{i\in F}a_i[j]-CB_j$, so $\sum_{i\in F}a_i[j]=B_j$. Therefore $\sum_{i\in D}a_i[j]=A_j-B_j=t[j]$
for every coordinate $j$, and $D$ is an RMSS solution.

Finally, the construction preserves the claimed structural parameter. Let $X$ be the set of all anchors, so $|X|=2d(K+1)$. After deleting $X$, every connected component consists of one item core $F_i$ with some pendant connector vertices attached to $u_i$.  Such a component has treedepth at most three: take $u_i$ as root, $v_i$ as its child, every $h_i^r$ as a child of $v_i$, and every connector vertex as a child of $u_i$.  Placing all vertices of $X$ on a chain above these rooted forests gives $\operatorname{td}(G)\leq2d(K+1)+3$. Thus $\operatorname{td}(G)+k$ is bounded by a function of the RMSS parameter $d+K$.  Since all numbers in RMSS are unary, $L,P,p,q$ and all connector multiplicities are polynomially bounded in the input size, so the reduction is polynomial-time.  This proves the theorem.
\end{proof}

\begin{corollary}
\label{cor:tw-k-hard}
Unless $\mathsf{FPT}=\mathsf{W[1]}$, 
\textsc{Bounded-Density Vertex Deletion} is not fixed-parameter tractable parameterized by $\operatorname{tw}(G)+k$ when the target density is part of the input.
\end{corollary}

The proof also pinpoints the numerical source of the hardness.  The constructed target is always of the form $2-3/q$ and approaches $2$ as the arithmetic capacity grows.  This is consistent with the fixed-target tractability results of Raupach et al.~\cite{RaupachBreitkopfHerrmannNichterlein2026}: the reduction fundamentally requires the target density to vary with the input. \qed

\section{Conclusion and Open Problems}
\label{sec:conclusion}

We resolve three open structural directions for \textsc{Bounded-Density Vertex Deletion}. First, the previously open twin-cover case is fixed-parameter tractable for arbitrary input density. Second, distance to cluster admits a sharp classification with respect to the numerical threshold: BDVD is FPT parameterized by \(\operatorname{dc}(G)+\lfloor2\tau_\rho+1\rfloor\), but is \(\mathsf{W[1]}\)-hard parameterized by \(\operatorname{dc}(G)\) alone when the target density is part of the input, even for integral targets. In particular, distance to cluster alone is FPT throughout \(\tau_\rho<1\). Third, the combined treewidth-plus-solution-size question has a negative answer: BDVD is \(\mathsf{W[1]}\)-hard already for \(\operatorname{td}(G)+k\), and hence also for \(\operatorname{tw}(G)+k\), with targets restricted to \((1,2)\).

A natural remaining structural question is the feedback edge number, explicitly left open by Raupach et al.~\cite{RaupachBreitkopfHerrmannNichterlein2026}. It is also interesting to seek purely combinatorial alternatives to the fixed-dimensional integer programs used for twin cover and distance to cluster and to investigate kernelization for these parameters.

\section*{Acknowledgments}

The authors used GPT-5.6 Sol during the preparation of this
work to brainstorm and refine constructions, explore proof
strategies, and discuss and check mathematical arguments.
The model also assisted with proofreading and aspects of
the literature review. All mathematical claims, proofs,
and cited results were independently verified by the authors,
who take full responsibility for the accuracy and content
of the paper.

\bibliography{references}
\end{document}